\documentclass[aps,pra,groupedaddress,superscriptaddress,amsmath,amssymb,showpacs,twocolumn,footinbib,floatfix,notitlepage]{revtex4-1}
\usepackage{graphicx}
\usepackage{bm,braket,dsfont}
\usepackage{color}
\usepackage{mathtools,amsthm}
\usepackage{hyperref}
\usepackage{microtype}

\usepackage{newpxtext,newpxmath}

\usepackage{multirow,array}
\usepackage{mathrsfs}

\let\mathbbm\mathds

\renewcommand{\Re}{\mathop{\text{Re}}\nolimits}

\newcommand{\cnot}{\textsc{cnot}}

\DeclareMathOperator{\Tr}{Tr}

\newcommand{\incoh}[1]{{\mathcal{C}^{(#1)}}}
\newcommand{\produc}[1]{{\mathcal{P}^{(#1)}}}
\newcommand{\two}[1]{\underline{2^{d-#1}}}

\renewcommand{\H}{\mathcal{H}}
\newcommand{\Hanc}{\mathcal{H}_{\text{anc}}}
\newcommand{\psianc}{\psi_{\text{anc}}}

\DeclareMathOperator{\SR}{R_{S}}
\DeclareMathOperator{\CR}{R_{C}}
\DeclareMathOperator{\SN}{N_{S}}
\DeclareMathOperator{\CN}{N_{C}}
\DeclareMathOperator{\NR}{R_{N}}
\DeclareMathOperator{\NN}{N_{N}}
\DeclareMathOperator{\ed}{D_{E}}

\newtheorem{theorem}{Theorem}
\newtheorem{lemma}[theorem]{Lemma}
\newtheorem{proposition}[theorem]{Proposition}

\newtheorem{definition}[theorem]{Definition}

\newcommand{\phin}{\phi^{(n)}}
\newcommand{\rhon}{\rho^{(n)}}

\renewcommand{\Lambda}{V}
\renewcommand{\Gamma}{W}

\newcommand\proj[1]{\ket{#1}\bra{#1}}

\newlength\negwidth

\makeatletter

  \adddialect\l@ENGLISH\l@english
    \adddialect\l@en\l@english
\makeatother

\begin{document}

\vspace*{2pt}

\title{Converting multilevel nonclassicality into genuine multipartite entanglement}

\author{Bartosz Regula}
\affiliation{School of Mathematical Sciences and Centre for the Mathematics and Theoretical Physics of Quantum Non-Equilibrium Systems, University of Nottingham, University Park, Nottingham NG7 2RD, United Kingdom}

\author{Marco Piani}
\affiliation{SUPA and Department of Physics, University of Strathclyde, Glasgow G4 0NG, UK}

\author{Marco Cianciaruso}
\affiliation{School of Mathematical Sciences and Centre for the Mathematics and Theoretical Physics of Quantum Non-Equilibrium Systems, University of Nottingham, University Park, Nottingham NG7 2RD, United Kingdom}

\author{Thomas~R.~Bromley}
\affiliation{School of Mathematical Sciences and Centre for the Mathematics and Theoretical Physics of Quantum Non-Equilibrium Systems, University of Nottingham, University Park, Nottingham NG7 2RD, United Kingdom}

\author{Alexander Streltsov}
\affiliation{Faculty of Applied Physics and Mathematics, Gda\'{n}sk University of Technology, 80-233 Gda\'{n}sk, Poland}
\affiliation{National Quantum Information Center in Gda\'{n}sk, 81-824 Sopot, Poland}

\author{Gerardo Adesso}
\affiliation{School of Mathematical Sciences and Centre for the Mathematics and Theoretical Physics of Quantum Non-Equilibrium Systems, University of Nottingham, University Park, Nottingham NG7 2RD, United Kingdom}

\date{\today}

\begin{abstract}
Characterizing genuine quantum resources and determining operational rules for their manipulation are crucial steps to appraise possibilities and limitations of quantum technologies. Two such key resources are nonclassicality, manifested as quantum superposition between reference states of a single system, and entanglement, capturing quantum correlations among two or more subsystems. Here we present a general formalism for the conversion of nonclassicality into multipartite entanglement, showing that a faithful reversible transformation between the two resources is always possible within a precise resource-theoretic framework. Specializing to quantum coherence between the levels of a quantum system as an instance of nonclassicality, we introduce explicit protocols for such a mapping. We further show that the conversion relates multilevel coherence and multipartite entanglement not only qualitatively, but also quantitatively, restricting the amount of entanglement achievable in the process and in particular yielding an equality between the two resources when quantified by fidelity-based geometric measures.
\end{abstract}

\maketitle

\section{Introduction}
\enlargethispage{\baselineskip}

Signature features of the quantum world have been recently recognized as {\it resources} that can be harnessed for disruptive technologies \cite{dowling_2003}.
One such resource, embodying the {\it nonclassicality} of quantum mechanics, is the possibility for a quantum system to exist in a {\em superposition} of ``classical'' states. The latter are usually determined based on physical considerations; for instance, in continuous-variable systems they can be identified with the  Glauber-Sudarshan coherent states \cite{glauber_1963,sudarshan_1963}, while in discrete-variable systems they can be taken to form a reference orthonormal basis (e.g.~the energy eigenbasis), so that superposition manifests as quantum {\em coherence} \cite{aberg_2006,witt_2013,levi_2014,baumgratz_2014,bromley_2014,winter_2016,napoli_2016,chitambar_2016,streltsov_2016}.

Superposition underlies other nonclassical phenomena such as quantum correlations among parts of a quantum system \cite{horodecki_2009,adesso_2016}. In particular, {\it entanglement} is itself a key resource and a characteristic trait of quantum mechanics, and stems from the superposition principle in conjunction with the tensor product structure associated to composite systems. Despite the common origin, entanglement and superposition can be formalized according to different resource theories: the former being tied to the paradigm of spatially separated laboratories which can only implement local operations and classical communication (\textsc{locc}) for free \cite{horodecki_2009}, while the second specified by the inability to create superpositions of the classical states for free \cite{aberg_2006,baumgratz_2014,killoran_2016,streltsov_2016,theurer_2017}. Consequently, these two resources, like two  currencies, enjoy different uses in quantum technologies. It thus becomes particularly relevant to investigate the connection between these two types of resource beyond a merely conceptual standpoint, and to devise operational schemes that allow the dynamical transformation of one into the other.

Several works have analyzed this problem.  In quantum optics, nonclassicality gets mapped into entanglement by a beam splitter \cite{kim_2002,wang_2002,wolf_2003,asboth_2005,ivan_2011}, while, in the discrete-variable scenario, it is the controlled \textsc{not} (\cnot) gate \cite{streltsov_2011,piani_2011} that plays a similar role. The quantitative interplay between the degree of nonclassicality and the bipartite entanglement obtained from it  has been investigated as well \cite{asboth_2005,sperling_2015,vogel_2014,streltsov_2015,ge_2015,killoran_2016,theurer_2017}. These studies have advanced our understanding of nonclassicality as a resource in systems of arbitrary dimension \cite{brandao_2008,baumgratz_2014,killoran_2016,zhang_2016,streltsov_2016,xu_2016,tan_2017,theurer_2017,mukhopadhyay_2017 }.

In this paper, we investigate the conversion of nonclassicality, expressed as superposition between multiple levels of a quantum system, into {\em multipartite} entanglement. In Sec.~\ref{sec:noncl} we show that there always exists a state-independent unitary mapping, realized by operations which alone cannot create nonclassicality, such that the presence of $k$-level nonclassicality in the state of a single $d$-level system is necessary and sufficient to create $k+1$-partite entanglement between the system and $k$ ancillas.

To exemplify such a conversion procedure, in Sec.~\ref{sec:coherence} we specialize to quantum coherence as an instance of nonclassicality \cite{streltsov_2016}, and introduce an explicit physical protocol which directly converts $k$-level coherence into $k+1$-body multipartite entanglement. The protocol entangles a $d$-level system (qudit) with up to $d$ qubits by a sequential application of generalized \cnot{} gates (free operations in the resource theory of coherence formalized in \cite{baumgratz_2014,winter_2016}). The protocol can be further extended via the decoupling of the qudit system by \textsc{locc} (free operations in the resource theory of entanglement \cite{horodecki_2009}), to provide a mapping of $k$-coherence into multipartite entanglement of the ancillary qubits alone. This process can also be seen as a toy model for decoherence~\cite{zurek_2003} due to the interaction with a many-body environment, with information about the superposition leaking into the environment in the form of multipartite entanglement.

Finally, in Sec.~\ref{sec:quantification} we show that the initial amount of $k$-coherence places a quantitative restriction on the amount of entanglement that can be converted from it. In particular, the fidelity-based geometric measure of $k+1$-partite entanglement \cite{wei_2003,streltsov_2010} at the output of the protocol is exactly equal to the fidelity-based geometric measure of $k$-coherence in the input state of the $d$-level system --- a computable quantifier of multi-level coherence introduced here, extending previous work in \cite{baumgratz_2014,streltsov_2015}.

\section{Nonclassicality conversion}\label{sec:noncl}

Nonclassicality is a notion that depends on the preassigned set of states that are deemed ``classical''. Choosing a finite set of states $\{\ket{\chi_i}\}$ which spans the whole Hilbert space $\H$ to constitute the pure classical states, as dictated by the physics of the problem under investigation, one asks whether a mixed state $\rho$ can be represented as a convex combination of classical states only. If this is not possible --- that is, if one has to consider superpositions of $\{\ket{\chi_i}\}$ --- then $\rho$ is a nonclassical state. In other terms, the set of all classical states $\mathcal{C}$ is formed by the convex hull of $\{\ket{\chi_i}\}$.

For finite-dimensional systems, the notion of nonclassicality is often understood as quantum coherence \cite{aberg_2006,baumgratz_2014,streltsov_2016}.  However, following \cite{killoran_2016,theurer_2017}, we note that the approach presented here is more general, since one does not require the states $\{\ket{\chi_i}\}$ to be orthogonal. This provides a common framework applicable e.g.~to classical sets formed by  $\operatorname{SU}(N)$ Gilmore-Perelomov coherent states \cite{gilmore_1972,perelomov_1972} in discrete-variable systems and Glauber-Sudarshan coherent states \cite{glauber_1963,sudarshan_1963} in continuous-variable systems, which are not orthogonal yet such that any finite subset thereof is linearly independent \cite{vogel_2014}.

The  framework adopted here leads to a natural measure of the level of nonclassicality of a state. For a pure state, one can indeed define the \textit{nonclassical rank} ($\NR$) \cite{sperling_2015,killoran_2016} as
$\mbox{$\NR \left(\ket\psi\right) = \min \left\{ r \;\big|\; \ket\psi = \sum_{i=1}^{r} c_i \ket{\chi_i},\; \ket{\chi_i} \in \mathcal{C} \right\}$}$,
with nonzero complex coefficients $c_i$. This clearly resembles the definition of the Schmidt rank $\SR(\ket\psi)$ of bipartite entangled states \cite{nielsen_2011}, and it can be extended to mixed states in the same way as the the Schmidt rank is extended to the Schmidt number $\SN$ \cite{terhal_2000}.

We thus define the \textit{nonclassical number} ($\NN$) of a mixed state $\rho$ as $\NN \left(\rho\right) = \min_{\{p_i, \ket{\psi_i}\}} \max_{i} \NR\left(\ket{\psi_i}\right)$, where the minimization is performed over all pure-state convex decompositions of $\rho$ into $\rho = \sum_i p_i \ket{\psi_i}\bra{\psi_i}$. In other words, in every such decomposition at least one state has nonclassical rank $\NR(\ket{\psi_i}) \geq \NN(\rho)$, and there exists a decomposition where all pure states have nonclassical rank $\NR(\ket{\psi_i}) \leq \NN(\rho)$.

Killoran et al.~\cite{killoran_2016} showed that there always exists an isometry, consisting of adding an ancilla and applying a global unitary, which maps each pure state of nonclassical rank $k$ into a bipartite entangled pure state of Schmidt rank $k$. In fact, as we show below, this result can be straightforwardly extended to the general case of mixed states:
\begin{proposition}\label{thm:map_to_schmidt}
Let $\H$ be a $d$-dimensional Hilbert space, $\mathcal{D}(\H)$ the corresponding set of density operators, and $\Hanc\cong\H$ the Hilbert space of an ancillary system. Then if the classical pure states $\{\ket{\chi_i}\}_{i=1}^d$ form a linearly independent set spanning $\H$, there exists an isometry $\Gamma: \H \rightarrow \H \otimes \Hanc$ such that for any state $\rho \in \mathcal{D}(\H)$ we have $\NN(\rho) = \SN(\Gamma\rho\Gamma^\dagger)$.
\end{proposition}
\begin{proof}To begin, let us note the fact that the set of all possible pure states belonging to pure-state decompositions of $\rho \in \mathcal{D}(\H)$ is given precisely by the set of pure states in the support of $\rho$ \cite{schrodinger_1936,hughston_1993}. By the result of \cite{killoran_2016}, we have that there exists a unitary $U$ such that $\NR(\ket\psi) = \SR\left(U(\ket\psi\otimes\ket\psianc)\right)\;\forall \ket\psi\in\H$ where $\ket\psianc \in \Hanc$ is a fixed reference state for the ancilla system. The isometry $W$ is given by attaching the ancilla state $\ket\psianc$ composed with the action of the unitary $U$. Following \cite{asboth_2005,piani_2012}, we note that there is a one-to-one correspondence between the pure-state decompositions of $\rho$ and the decompositions of $\rho' = W\rho W^\dagger$, given exactly by the action of $W$. Notice in particular that $W$ can be inverted on its image, and that any $\ket{\psi'}$ in the support of $\rho'$ has as pre-image $W^\dagger \ket{\psi'}$ in the support of $\rho$, which means that by the properties of $W$ one has $\NR(W^\dagger \ket{\psi'}) = \SR\left(\ket{\psi'}\right)$. Assuming $\rho$ has $\NN(\rho) = k$, then it is possible to find a pure-state decomposition of it which only contains states with non-classicality rank less or equal to $k$. The pure states in one such decomposition will then be transformed by the action of $W$ into entangled states of Schmidt rank at most $k$, which will form a pure-state decomposition of $\rho'$. This proves that $\SN(\rho')\leq \NN(\rho)$. On the other hand, suppose $\SN(\rho') = l$; then it is possible to find a pure-state decomposition of $\rho'$ such that it only contains states with Schmidt rank less or equal to $l$. Under the action of $W^\dagger$, such a decomposition gives rise to a pure-state decomposition of  $\rho$ whose elements have non-classicality rank at most $l$. This proves that $\NN(\rho)\leq \SN(\rho')$.
\end{proof}


 In this paper, we show that an analogous faithful conversion of multilevel nonclassicality into genuine multipartite entanglement is always possible. Following \cite{guhne_2005}, we define a pure state $\ket\psi$ to be $k$-producible if it can be written as $\ket\psi = \ket{\psi_1} \otimes \ldots \otimes \ket{\psi_m}$ with each $\ket{\psi_j}$ pertaining to at most $k$ parties, and a mixed state $\rho$ to be $k$-producible if it can be written as a convex combination of $k$-producible pure states. We call a state $\rho$ genuinely $k$-partite entangled if it is $k$-producible but not $k-1$-producible; equivalently, under such conditions we say that $\rho$ has {\it entanglement depth} $\ed(\rho) = k$ \cite{sorensenmolmer}. A $1$-producible state $\rho$ has $\ed(\rho)=1$ and is fully separable.

\begin{theorem}\label{thm:map_to_k1}
Let $\H$ be a $d$-dimensional Hilbert space, and $\Hanc$ the Hilbert space of an ancillary system. Then if the classical pure states $\{\ket {\chi_i}\}_{i=1}^{d} \in \H$ form a linearly independent set spanning $\H$, there exists an isometry $\Lambda: \H \to \H \otimes \Hanc^{\otimes d}$ such that for any state $\rho \in \mathcal{D}(\H)$ with nonclassical number $\NN(\rho) = k$, $\Lambda\rho\Lambda^\dagger$ is genuinely $k+1$-partite entangled iff $\rho$ is nonclassical ($2\leq k \leq d$) and $\Lambda\rho\Lambda^\dagger$ is fully separable iff $\rho$ is classical ($k=1$).\end{theorem}

\begin{proof}We adapt the methods of Ref. \cite{killoran_2016} to show the existence of this mapping. Let us consider the case of pure states first. Define
\begin{equation}\ket{c_i} = \ket{\chi_i} \otimes \ket\psianc\in\H\otimes\Hanc^{\otimes d}\end{equation}
with $\ket\psianc \in \Hanc^{\otimes d}$ a fixed (fully unentangled) reference state of the ancilla systems. Define $\{\ket{b_i}\}_{i=1}^{d} \in \Hanc^{\otimes d}$ as
\begin{equation}\ket{b_i} = \ket{0}^{\otimes i-1} \otimes \ket\lambda \otimes \ket{0}^{\otimes d-i} \end{equation}
with $\ket\lambda = \sqrt{\lambda} \ket{0} + \sqrt{1-\lambda} \ket{1}$, where $\ket{0},\ket{1}\in \Hanc$ are orthonormal, and $0 \leq \lambda \leq 1$. Recall that the Gram matrix of a set of states $\{\ket{\phi_i}\}$ is defined as $[G^{(\phi)}]_{ij} = \braket{\phi_i | \phi_j}$, and has full rank iff the $\ket{\phi_i}'s$ are linearly independent ~\cite{jozsa_2000}. Define a $\mu$-dependent matrix $B(\mu)$ such that $[B(\mu)]_{ij} = 1$ if $i=j$, $\mu$ if $i\neq j$. We have $G^{(b)}=B(\lambda)$. Define $M(\epsilon) = G^{(c)} \circ B(1+\epsilon)$ where $\circ$ is the Hadamard (that is, entrywise) product. Since $\lim_{\epsilon\to 0}M(\epsilon) = G^{(c)} > 0$ and $\operatorname{diag}(M(\epsilon))=(1,1,\ldots,1)$, it follows that, for sufficiently small $\epsilon>0$, $M(\epsilon)$ is the Gram matrix $G^{(a)}$ of a  set of linearly independent states $\{\ket{a_i}\}_{i=1}^{d}$. We then have $G^{(c)} = G^{(a)} \circ B\left(\frac{1}{1+\epsilon}\right) = G^{(a)}\circ G^{(b)}$, for $\lambda = (1+\epsilon)^{-1}$, which means that the sets of states $\{\ket{c_i}\}$ and $\{\ket{a_i}\otimes\ket{b_i}\}$ have equal Gram matrices, and so there exists a unitary $U$ such that $U \ket{c_i} = \ket{a_i}\otimes\ket{b_i}\,\forall\, i$ \cite{jozsa_2000,chefles_2004}. The isometry $V$ is defined by the composition of attaching the ancilla state $\ket\psianc$ followed by the action of $U$.

Now consider a general pure qudit state $\ket{\psi} =  \sum_{i=1}^{d} \psi_i \ket{\chi_i} $. Then,
\begin{equation}\begin{aligned}\ket{\psi'} &= V\ket\psi = \sum_{i=1}^{d} \psi_i \ket{a_i} \ket{b_i}\\
&= \sum_{i=1}^{d} \psi_i \ket{a_i} \ket{0}^{\otimes i-1} \ket\lambda \ket{0}^{\otimes d-i}.\end{aligned}\end{equation}
It is convenient to use the fact that the entanglement depth of $\ket{\psi'}$ is not affected by a local filter $S\otimes L^{\otimes d}$, with $S$ a qudit operator such that $S\ket{a_i} = \ket{i}$, and $L$ a qubit operator such that $L\ket{0} = \ket{0}$, $L\ket{\lambda} = \ket{1}$. Thus, we can study the entanglement depth of the state
\begin{equation}\ket{\widetilde{\psi}'}\propto \sum_{i=1}^{d} \psi_i \ket{i} \ket{0}^{\otimes i-1} \ket{1} \ket{0}^{\otimes d-i} = \sum_{i=1}^{d} \psi_i \ket{i} \ket{\underline{2^{d-i}}},\end{equation}
where $\ket{\underline{2^{d-i}}}$ is the string of qubits corresponding to $2^{d-i}$ in binary (padded with zeros from the left as needed), e.g. $\ket{\underline{2^3}} = \ket{00\cdots 01000}$ since $2^3 = 1000_2$.
It is evident that $\ket{\widetilde\psi'}$ is fully product iff there is only one term in the superposition, that is $\NR(\ket\psi) = 1 $ iff $\ed(\ket{\psi'}) = 1$. In the following we will consider $\NR(\ket\psi)\geq 2$, and in this case we will prove that $\ed(\ket{\psi'}) = \NR(\ket\psi) +1 $.

To show that $\NR(\ket\psi) = k \geq 2$ implies $\ed(\ket{\widetilde\psi'}) = k +1$, assume w.l.o.g.~that the  $k$ nonzero coefficients $\psi_i$ are the first ones. Then $\ket{\widetilde{\psi}'} \propto \big(\sum_{i=1}^{k} \psi_i \ket{i} \ket{0}^{\otimes i-1} \ket{1} \ket{0}^{\otimes k-i}\big)\ket{0}^{\otimes d-k}$, and the claim follows by showing that $\sum_{i=1}^{k} \psi_i \ket{i} \ket{0}^{\otimes i-1} \ket{1} \ket{0}^{\otimes k-i}$ cannot be factorized in any non-trivial way. This holds, as the reduced state of the $k$ ancillary qubits is proportional to $\sum_{i=1}^k |\psi_i|^2 \proj{0}^{\otimes i-1} \otimes \proj{1} \otimes \proj{0}^{\otimes k-i}$, so the marginal state of any subset of these $k$ qubits is evidently mixed.
On the other hand, to prove that  $\ed(\ket{\psi'}) = \ed(V\ket{\psi})= k +1 \geq 3$ implies $\NR(\ket\psi) = k$, notice that the isometry is invertible on its image. Thus, $\ket\psi = V^\dagger \ket{\psi'}$; since we have just proven that $\ed(V\ket{\psi}) = \NR(\ket\psi)+1$, we arrive at the claim for pure states.

The mixed-state case follows \cite{asboth_2005,piani_2012} by noting that there is a one-to-one correspondence between the pure state decompositions $\{p_i,\ket{\psi_i}\}$ of $\rho$ and $\{p_i,\ket{\psi_i'}\}$ of $\rho' = V \rho V^\dagger$, with each input-output pair of states respecting the relation just discussed: either $\NR(\ket{\psi_i}) = \ed(\ket{\psi_i'}) = 1$  or $\NR(\ket{\psi_i}) +1 = \ed(\ket{\psi_i'})$. Thus, with the exception of the first (trivial) case,
 we have $\NR(\rho) +1 = \ed(\rho')$. Indeed, the pure-state mapping with the above properties, together with the definitions of nonclassical number and entanglement depth, entail that  $\NN(\rho)= m $ implies $\ed(\rho') \leq m + 1 $, and, in turn, $\ed(\rho') =  l $ implies $\NN(\rho') \leq l - 1$. These relations can only be satisfied for $l = m+1$.
\end{proof}

Theorem \ref{thm:map_to_k1} shows that there always exists an isometry which faithfully converts the $k$-level nonclassicality of a quantum system into multipartite entanglement with $k$ other ancillary systems. We note that the specifics of the mappings are not fixed by the theorem, and one could always devise other ways to convert the nonclassicality into entanglement. In particular, the mappings presented in the proofs only use two levels of each ancillary system, resulting in entanglement akin to that of W states \cite{dur_2000}. One may consider other kinds of operations which create qualitatively different multipartite entanglement  --- for instance, one can instead attach a number $d$ of $d+1$-dimensional ancilla systems and choose $\{\ket{b_j}\}_{j=1}^{d}$ such that $\ket{b_j} = \ket{\lambda_j}^{\otimes d}$ and $\ket{\lambda_{j}} = \sqrt{\lambda^{1/d}} \ket{0} + \sqrt{1-\lambda^{1/d}} \ket{j}$ where $\{\ket{j}\}$ is now an orthonormal basis for $\Hanc$. Following a similar argument to Thm.~\ref{thm:map_to_k1}, this will then introduce a generalized GHZ-type entanglement between the qudits, entangling as many levels of the systems as the nonclassical rank of the original state. However, the choice of a W-type mapping in the theorems makes the conversion quite appealing in practice, as it only requires qubit ancillas, and, as we show below, enables one to create entanglement by a sequential application of two-body gates on the nonclassical system and each ancilla.

\section{Coherence conversion protocol}\label{sec:coherence}

We will now specialize to the framework of quantum coherence \cite{aberg_2006,baumgratz_2014,winter_2016,streltsov_2016}. Here, the classical states $\{\ket i\}_{i=1}^{d}$ are taken to form a fixed orthonormal basis for $\H$. Analogously to nonclassicality, we can then define a hierarchy of coherence levels by considering the \textit{coherence rank} $\CR(\ket\psi)$, defined to be the number of nonzero coefficients $c_i$ that a state $\ket\psi = \sum_i c_i \ket{i}$ has in this basis \cite{witt_2013,levi_2014}. We then define the \textit{coherence number}
\begin{equation}\CN(\rho)=\min_{\{p_i, \ket{\psi_i}\}} \max_{i}\CR\left(\ket{\psi_i}\right)\end{equation}
for a mixed state $\rho$ accordingly. We will refer to states with coherence number $k$ as {\it $k$-coherent states}. Clearly, 1-coherence corresponds to classicality, $k$-coherence for any $k\geq 2$ stands as a fine graining of the usual notion of coherence, and $d$-coherence is the maximal coherence level of a $d$-level system.

The $k$-coherence of a single qudit  can be converted into multipartite entanglement in  different physical ways. To show this, we design a protocol to convert $k$-coherence into $k+1$-partite entanglement between the qudit and $k$ qubit ancillas (following Thm.~\ref{thm:map_to_k1}), realizable by a sequential application of $\cnot{}$ gates (see Fig.~\ref{fig:protocol}). We then provide a natural mapping of $k$-coherence into $k$-body entanglement, which can be accomplished by a second step which disentangles the qudit system --- either by unitary transformations as in Fig.~\ref{fig:protocol}(a), or by one-way \textsc{locc}
as in Fig.~\ref{fig:protocol}(b). The latter scheme reflects an operational  scenario in which input agents are constrained to the resource theory of $k$-coherence, having at disposal only incoherent ancillas and incoherent operations as used in the first step, while output agents are constrained to the resource theory of entanglement, being bound to use \textsc{locc} as in the second step.

We illustrate the scheme for pure states, noting that it extends straightforwardly to mixed states. Let $\ket{\Psi} = \ket{\psi^d} \otimes \ket{0}^{\otimes d}$ be the state of the composite system consisting of the qudit initialized in $\ket{\psi^d}$ and $d$ ancilla qubits in a reference pure state $\ket{0}^{\otimes d}$. Consider a unitary activation operation $U_A$ which consists of a sequence of generalized $\cnot{}$ gates $(\mathbbm{1}_d-\proj{i})\otimes \mathbbm{1}_2 + \proj{i} \otimes \sigma_x$, with $\sigma_x$ the Pauli $x$ matrix, between the qudit and the $i$-th ancillary qubit. Explicitly, the sequence realizes the unitary
\begin{equation}U_A = \sum_{i=1}^{d} \ket{i}\bra{i} \otimes \mathbbm{1}_{2}^{\otimes i-1} \otimes \sigma_x \otimes \mathbbm{1}_{2}^{\otimes d-i},\end{equation}
which transforms the state $\ket{\Psi} = \sum_{i=1}^{d} c_i \ket{i} \ket{0}^{\otimes d}$ into $\ket{\Psi'} = U_A \ket{\Psi}=\sum_{i=1}^{d} c_i \ket{i} \ket{\underline{2^{d-i}}}$.

To complete the protocol by mapping into $k$-partite entanglement among the qubit ancillas only, we now give two alternative approaches. Both methods begin by performing a quantum Fourier transform (\textsc{qft}) $\ket{j} \rightarrow \frac{1}{\sqrt{d}} \sum_m e^{\,2\pi i j k / d} \ket{m}$ on the qudit only. Then, in the first approach [Fig.~\ref{fig:protocol}(a)], we can apply a unitary
\begin{equation}U_D = \sum_{j,m=1}^{d} e^{-2\pi i j m / d} \ket{m}\bra{m} \otimes \ket{\underline{2^{d-j}}}\bra{\underline{2^{d-j}}}\end{equation}
to effectively decouple the qudit and the ancilla qubits. This can be understood as the sequential application of $d^2$ controlled local operations $(\mathbbm{1}_d-\proj{m})\otimes \mathbbm{1}_2 + \proj{m} \otimes U^{(m)}_{D_j}$, with control on the qudit and
\begin{equation}U^{(m)}_{D_j} = \ket{0}\bra{0} + e^{-2 \pi i j m / d}\proj{1}\end{equation}
acting on the $j$-th ancillary qubit. After the action of the \textsc{qft} and $U_D$, which jointly define the global unitary $U_B$, the output will be the product state $U_B \ket{\Psi'}=\ket{\Phi^{+}}\ket{\Psi''}$, where $\ket{\Phi^+} = \sum_i \frac{1}{\sqrt{d}} \ket{i}$ is the maximally coherent state of the qudit, and $\ket{\Psi''} = \sum_{i=1}^{d} c_i \ket{\underline{2^{d-i}}}$ is a $k$-partite entangled state of the qubit ancillas.

An alternative approach [Fig.~\ref{fig:protocol}(b)], which might lend itself to a more efficient implementation as it does not require global interactions, is to realize the decoupling of the qudit by an operation $\Delta$ consisting of one-way \textsc{locc} (see e.g.~Ref.~\cite{piani_2012}). After performing the \textsc{qft}, one can measure the qudit in the $\{\ket m\}$ basis and, depending on the measurement result $m$, apply the local unitary $U_{D}^{(m)} = \sum_{j=1}^{d} e^{-2\pi i j m / d} \ket{\underline{2^{d-j}}}\bra{\underline{2^{d-j}}} = \bigotimes_j U^{(m)}_{D_j}$ to the remaining $d$ qubits. We then obtain the final state $\Delta\left(\ket{\Psi'}\right) =\ket{\Psi''}$, which is exactly the same as the state of the qubits after the  unitary transformation $U_B$ from the previous approach.

We can formalize the properties of the protocol as follows, casting the result in terms of mixed states in general.

\begin{figure*}[t]
\vspace*{-10pt}
\begin{minipage}[b]{7.2cm}
\includegraphics[width=7.2cm]{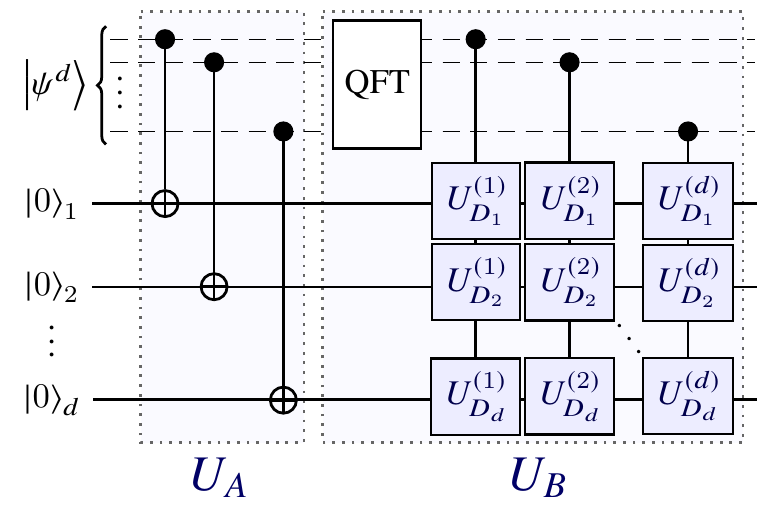} \\
{\small{(a)}}
\end{minipage}%
\hspace*{1cm}%
\begin{minipage}[b]{6.3cm}
\includegraphics[width=6.3cm]{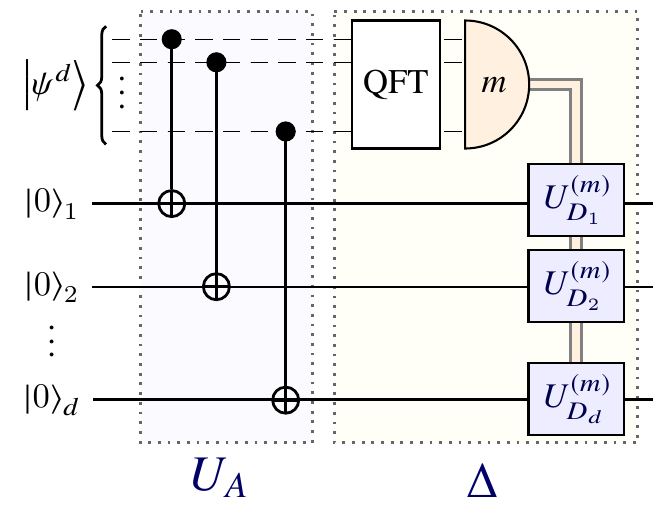} \\
{\small{(b)}}
\end{minipage}
\caption{Schemes of two protocols to convert $k$-coherence into multipartite entanglement. Both protocols begin with the global unitary operation $U_A$ which sequentially entangles each level of the qudit system in the state $\ket{\psi^d}$ with a corresponding ancillary qubit by generalized \cnot{} gates, resulting in a $k+1$-partite entangled state. One can then decouple the qudit system either (a) by a unitary transformation $U_B$, consisting of a Fourier transform and a disentangling unitary $U_D$, or (b) via a one-way \textsc{locc} operation $\Delta$. Both protocols result in genuine $k$-partite entanglement between the ancillary qubits.}
\label{fig:protocol}
\vspace*{-1pt}
\end{figure*}

\begin{theorem}
\label{thm:coherence_pure}
Given the above  conversion protocol, consisting of the activation unitary $U_A$ and the decoupling operation (either via $U_B$ or $\Delta$), with $\rho'=U_A (\rho\otimes \proj{0}^{\otimes d})U_A^\dagger$ and $\rho'' = \Delta(\rho')$, the following statements are equivalent for $2 \leq k \leq d$:
\renewcommand{\labelenumi}{(\roman{enumi})}
\begin{enumerate}
\item $\CN(\rho)=k$;
\item $\ed(\rho')=k+1$;
\item $\ed(\rho'') = k$.
\end{enumerate}
\end{theorem}
\begin{proof}
The equivalence between (i) and (ii) is a direct application of Thm.~\ref{thm:map_to_k1}. To prove that (i) implies (iii), we first consider pure states, and assume w.l.o.g. that the $k$ nonzero coefficients are the first ones. Then the last $d-k$ qubits in $\ket{\Psi''}$ are in the initial product state $\ket{0}^{\otimes d-k}$, and we need to prove that the first $k$ qubits are genuinely multipartite entangled. This is the case since any non-trivial subset of such qubits is mixed. That (iii) implies (i) can then be proven by observing, as in the proof of Thm.~\ref{thm:map_to_k1}, that the isometry from $\ket{\psi^d}$ to $\ket{\Psi''}$ can be inverted, and using the just proven fact that (i) implies (ii). The extension to mixed states follows the exact same steps as in Thm.~\ref{thm:map_to_k1}.
\end{proof}
We notice that (i) and (iii) are actually also equivalent for $k=1$, which does not hold in the case of (i) and (ii) since classicality does not lead to entanglement creation.

\section{Quantitative relations}\label{sec:quantification}

In any resource theory, one can define a faithful class of quantifiers by considering the distance to the set of non-resource states \cite{bengtsson_2007,adesso_2016,brandao_2015}. In the cases of bipartite entanglement and standard coherence (i.e., $2$-coherence in our framework), the corresponding non-resource sets are the sets of separable states $\mathcal{S}$ and incoherent states $\mathcal{I}$, respectively \cite{vedral_1997,baumgratz_2014}. For the case of $k$-partite entanglement, one can define the non-resource set as the set of $k-1$-producible states $\produc{k-1}$ \cite{guhne_2005}, i.e., states which are at most $k-1$-partite entangled. Similarly for $k$-coherence, we consider the set $\incoh{k-1}$ of states which are at most $k-1$-coherent. We then define the distance-based quantifiers as follows.
\begin{definition}
Given a quasi-metric $D(\rho,\sigma)$ contractive under \textsc{cptp} maps, we define the distance-based measure of $k$-partite entanglement as
\begin{equation}
E_{D}^{(k)}(\rho) = \inf_{\varsigma \in \mathcal{P}^{(k-1)}} D(\rho, \varsigma)
\end{equation}
and the distance-based measure of $k$-coherence as
\begin{equation}
C_{D}^{(k)}(\rho) = \inf_{\sigma \in \incoh{k-1}} D(\rho, \sigma).
\end{equation}
\end{definition}
The distance-based quantifiers of $k$-coherence and $k$-partite entanglement have many useful properties which allow us to relate the two resource quantitatively. In particular, within the distance-based framework, we prove the following relation between the degree of coherence of a state $\rho$ and the multipartite entanglement of the output states $\rho' = U_A (\rho\otimes \proj{0}^{\otimes d})U_A^\dagger$ and $\rho'' = \Delta(\rho')$ obtained from the conversion protocol of Thm.~\ref{thm:coherence_pure}:
\begin{theorem}
\label{thm:distance_based_relations}
Let $D$ be any distance contractive under \textsc{cptp} maps. Then
\begin{align*}&C_D^{(k)}(\rho) \geq E_D^{(k+1)}(\rho'), \\
&C_D^{(k)}(\rho) \geq E_D^{(k)}(\rho'').\end{align*}
\end{theorem}
\begin{proof}Let $\mathcal{B}^{(k)}$ denote the subset of $\mathcal{P}^{(k)}$ spanned by states of the form $\ket{\Psi'} = \sum_{i=1}^{d} c_i \ket{i} \ket{\underline{2^{d-i}}}$ as obtained from the first step of the conversion protocol in Thm.~\ref{thm:coherence_pure}. Similarly, let $\mathcal{R}^{(k-1)}$ denote the subset of $\mathcal{P}^{(k-1)}$ spanned by $\ket{\Psi''} = \sum_{i=1}^{d} c_i \ket{\underline{2^{d-i}}}$ as obtained from the second step of the same protocol.

Let us consider $\rho'$ first. We have that
\begin{equation}
\begin{aligned}
C_{D}^{(k)}(\rho) &= \hspace{-5pt} \inf_{\sigma \in \incoh{k-1}} D(\rho, \sigma) \\
&= \hspace{-5pt}\inf_{\sigma \in \incoh{k-1}} D(U_A \rho \otimes \ket{0}\bra{0}^{\otimes d} U_A^{\dagger}, U_A \sigma \otimes \ket{0}\bra{0}^{\otimes d} U_A^{\dagger})  \\
&= \hspace{-5pt}\inf_{\sigma \in \incoh{k-1}} D(\rho', U_A \sigma \otimes \ket{0}\bra{0}^{\otimes d} U_A^{\dagger})  \\
&= \hspace{-5pt}\inf_{\delta \in \mathcal{B}^{(k)}} D(\rho', \delta) \\
&\geq \hspace{-5pt}\inf_{\varsigma \in \mathcal{P}^{(k)}} D(\rho', \varsigma)  \\
&= E_{D}^{(k+1)}(\rho'),
\end{aligned}
\end{equation}
where we have used Thm.~\ref{thm:coherence_pure},
as well as the facts that $D(\rho \otimes \tau,\sigma \otimes \tau)=D(\rho,\sigma)$ and $D(U\rho U^{\dagger},U \sigma U^{\dagger}) = D(\rho,\sigma)$ for any contractive distance $D$. Equality clearly holds when there exists $\delta \in \mathcal{B}^{(k)}$ such that
\begin{equation}
\inf_{\varsigma \in \mathcal{P}^{(k)}} D(\rho', \varsigma) = D(\rho', \delta).
\end{equation}

An analogous argument holds for $\rho''$, One can either follow the steps above with the unitary transformation $U_B$, or note the contractivity of the distance $D$ under the \textsc{locc} operation $\Delta$ and obtain:
\begin{equation}
\begin{aligned}
C_{D}^{(k)}(\rho) &= \inf_{\sigma \in \incoh{k-1}} D(\rho, \sigma) \\
&= \inf_{\sigma \in \incoh{k-1}} D(\rho', U \sigma \otimes \ket{0}\bra{0}^{\otimes d} U^{\dagger})  \\
&\geq \inf_{\sigma \in \incoh{k-1}} D\left(\Delta(\rho'), \Delta\left(U \sigma \otimes \ket{0}\bra{0}^{\otimes d} U^{\dagger}\right)\right) \\
&= \inf_{\delta \in \mathcal{R}^{(k-1)}} D(\rho'', \delta) \\
&\geq \inf_{\varsigma \in \mathcal{P}^{(k-1)}} D(\rho'', \varsigma)  \\
&= E_{D}^{(k)}(\rho'').
\end{aligned}
\end{equation}
\end{proof}
The amount of $k$-coherence present in the initial state thus places quantitative constraints on the multipartite entanglement one can obtain from it.

We can obtain a particularly interesting family of distance-based quantifiers by setting $D(\rho, \sigma) := 1- F(\rho, \sigma)$, with
\begin{equation}
F(\rho,\sigma)=\Tr\Big(\sqrt{\smash[b]{\sqrt{\rho}\sigma\sqrt{\rho}}\vphantom{\big(}}\Big)^2
\end{equation} being the (squared) fidelity \cite{uhlmann_1976,jozsa_1994}. These quantifiers are related to the family of {\it geometric measures} of $k$-coherence $C_G^{(k)}$ and $k$-partite entanglement $E_G^{(k)}$, which directly generalize their counterparts defined first for entanglement \cite{wei_2003,streltsov_2010} and standard quantum coherence \cite{baumgratz_2014,streltsov_2015}:
\begin{definition}
The geometric measure of $k$-partite entanglement is given by
\begin{equation}
\label{eq:ent_convexroof}\begin{aligned}
E_G^{(k)} (\ket\psi) &= \inf_{\ket\varsigma \in \mathcal{P}^{(k-1)}} \left(1 - F(\ket\varsigma, \ket\psi)\right)\\
E_G^{(k)} (\rho) &= \inf_{\{p_i, \ket{\psi_i}\}}  \sum_i p_i \, E_G^{(k)}(\ket{\psi_i})\\&\text{s.t.}\;\; \sum_i p_i \ket{\psi_i}\bra{\psi_i} = \rho,\; p_i\geq 0 \,\forall i,
\end{aligned}
\end{equation}
and similarly for the geometric measure of $k$-coherence:
\begin{equation}\begin{aligned}
C_G^{(k)} (\ket\psi) &= \inf_{\ket\sigma \in \incoh{k-1}} \left(1 - F(\ket\sigma, \ket\psi)\right)\\
C_G^{(k)} (\rho) &= \inf_{\{p_i, \ket{\psi_i}\}}  \sum_i p_i \,C_G^{(k)}(\ket{\psi_i})\\&\text{s.t.}\;\; \sum_i p_i \ket{\psi_i}\bra{\psi_i} = \rho,\; p_i\geq 0 \,\forall i
\end{aligned}
\end{equation}
where $F(\ket\sigma, \ket\psi) = | \braket{\sigma | \psi} |^2$ is the (squared) fidelity.
\end{definition}
In fact, the geometric measures and the fidelity-based distance quantifiers can be shown to be equal to each other, as we prove below.
\begin{proposition}
\label{prop:geometric_fidelity}
The following relation holds for the geometric measures:
\begin{equation}\begin{aligned}
E_G^{(k)}(\rho) = \inf_{\varsigma \in \mathcal{P}^{(k-1)}} \left(1 - F(\rho, \varsigma)\right)\\
C_G^{(k)}(\rho) = \inf_{\sigma \in \incoh{k-1}} \left(1 - F(\rho, \sigma)\right)
\end{aligned}\end{equation}
where $F(\rho,\sigma) = \left(\Tr \sqrt{\sqrt{\rho}\sigma\sqrt{\rho}} \right)^2$
is the (squared) fidelity.
\end{proposition}
\begin{proof}Let $X$ denote either $\mathcal{P}^{(k-1)}$ or $\incoh{k-1}$. We know that $X$ is a closed convex set with its extremal points $\operatorname{ext}(X)$ given by pure states, so by the result of Thm. 2 in the appendix of \cite{streltsov_2010} we have
\begin{equation}
\label{eq:mixed_convex_roof}
\max_{\sigma \in X} F(\rho, \sigma) = \max_{\{p_i, \rho_i\}}\sum_i p_i \, \max_{\ket{\delta}\in \operatorname{ext}(X)} F(\rho_i, \ket\delta \bra\delta)
\end{equation}
where the maximization is performed over all convex mixed-state decompositions of $\rho = \sum_i p_i \rho_i$. Let $\{p_i, \rho_i\}$ be the decomposition which realizes the first maximization on the right-hand side, and note that every such $\rho_i$ can in turn be expressed as a convex decomposition into pure states $\ket{\psi^i_j}$ as $\rho_i = \sum_j q^i_j \ket{\psi^i_j}\bra{\psi^i_j}$. Let $\{\ket{\delta_i}\}$ be the states which realize the second maximization for each $\rho_i$. We then get
\begin{equation}
\begin{aligned}
\max_{\sigma \in X} F(\rho, \sigma) &= \sum_i p_i \, \max_{\ket{\delta}\in \operatorname{ext}(X)} F(\rho_i, \ket\delta \bra\delta)\\
&= \sum_{i} p_i \,F(\rho_i, \ket{\delta_i}\bra{\delta_i})\\
&= \sum_{i,j} p_i q^i_j \braket{\delta_i | \psi^i_j} \braket{\psi^i_j | \delta_i}
\end{aligned}
\end{equation}
which shows that
the maximum in Eq.~\eqref{eq:mixed_convex_roof} is in fact always reached by a pure-state decomposition of $\rho = \sum_{i,j} p_i q^i_j \ket{\psi^i_j}\bra{\psi^i_j}$.
\end{proof}

Remarkably, under the geometric quantifiers, the $k$-coherence of any state and the converted $k+1$-partite entanglement are in fact \emph{equal}. The result relies on the following lemma, which shows that for any pure state $\ket{\Psi'}$ obtained from the conversion protocol, it suffices to optimize the distance-based quantifier of $k+1$-partite entanglement over the set of $k$-producible output states of the protocol, instead of the whole set of $k$-producible states.
\begin{lemma}
\label{lem:purefidelity}
Given a state of the form
\begin{equation}\label{eq:protocol}
\ket{\psi} = \sum_{i=1}^{d} c_i \ket{i}\ket{\underline{2^{d-i}}}
\end{equation}
where we can take $|c_1|\geq |c_2|\geq \ldots \geq |c_d|$ without loss of generality, the closest $k$-producible state with respect to the fidelity-based geometric measure of entanglement can be chosen as a state $\ket{\psi^c} \in \mathcal{B}^{(k)}$.

In other words, there exists $\ket{\psi^c} \in \mathcal{B}^{(k)}$ such that
\begin{equation}
\max_{\ket\varsigma \in \mathcal{P}^{(k)}} F(\ket{\psi},\ket{\varsigma}) = \max_{\ket\sigma \in \mathcal{B}^{(k)}} F(\ket\psi, \ket\sigma) = F(\ket{\psi},\ket{\psi^c}).
\end{equation}
\end{lemma}

\begin{proof}
Recall that, by Thm. 1, a state in $\mathcal{B}^{(k)}$ has $p$ nonzero coefficients ($2 \leq p \leq d$) iff it is $p+1$-partite entangled. Therefore, we have that any $\ket\omega \in \mathcal{B}^{(k)}$ is given by
\begin{equation}
\ket\omega = \sum_{j \in \mathcal{J}} d_j \ket{j}\ket{\two{j}}
\end{equation}
where $\mathcal{J}$ is a subset of at most $k-1$ elements of $\{1,\ldots,d\}$. Hence
\begin{equation}
\begin{aligned}
F(\ket\psi, \ket\omega) &= |\braket{\psi | \omega}|^2\\
&=  \left| \sum_{i=1}^{d} \sum_{j \in \mathcal{J}} \overline{c_i} d_j \braket{i|j} \braket{\two{i}|\two{j}} \right|^2\\
&= \left| \sum_{j\in \mathcal{J}} \overline{c_j} d_j \right|^2\\
&\leq \sum_{j\in \mathcal{J}} |c_j|^2 \sum_{j\in \mathcal{J}} |d_j|^2\\
&= \sum_{j\in \mathcal{J}} |c_j|^2\\
&\leq \sum_{i=1}^{k-1} |c_i|^2
\end{aligned}
\end{equation}
using the Cauchy-Schwarz inequality and the fact that the coefficients $|c_i|$ are arranged in non-increasing order, so the choice of the first $k-1$ coefficients maximizes the expression. This bound is saturated by the renormalised state given by
\begin{equation}
\ket{\psi^{c}} = \frac{\sum_{i=1}^{k-1} c_i \ket{i}\ket{\two{i}}}{\sqrt{\sum_{j=1}^{k-1} |c_i|^2}} \;\in\mathcal{B}^{(k)}
\end{equation}
and so it is tight.

Therefore, we need to show that for any $\ket\phi \in \mathcal{P}^{(k)}$ we have $|\braket{\psi | \phi}|^2 \leq \sum_{j=1}^{k-1} |c_i|^2$.

For a general $k$-producible state, we have to consider several different cases corresponding to different partitions. Consider a general $k$-producible state $\ket{\phin}$, where the qudit is entangled with $n$ qubits and the remaining $(d-n)$ qubits are in some arrangement --- we do not explicitly assume anything about the state of the $(d-n)$ qubits. The state has the following form:
\begin{equation}
\begin{aligned}
\ket{\phin} &= \ket{\text{qudit + } n \text{ qubits}}\otimes\ket{d-n \text{ qubits}}\\
&= \ket{\phin_A} \otimes \ket{\phin_B}.
\end{aligned}
\end{equation}
We necessarily have that $n \leq k-1$, because otherwise $\ket\phin$ would not be $k$-producible.

Note that for a general bipartite state $\ket{\alpha_{AB}}$, it follows from the Schmidt decomposition that \cite{shimony_1995,wei_2003}
\begin{equation}
\begin{aligned}
\max_{\text{separable }\ket{\eta_{AB}}} |\braket{\eta_{AB}|\alpha_{AB}}|^2 = \lambda_{\max} \Big(\Tr_A \big(\ket{\alpha_{AB}}\bra{\alpha_{AB}}\big)\Big)
\end{aligned}
\end{equation}
where $\lambda_{\max}$ denotes the largest eigenvalue. In our case, we can treat $\ket\psi$ as a bipartite state where the subsystem $A$ is comprised of the qudit and $n$ qubits, and $\ket\phin$ is the corresponding product state. We then have
\begin{equation}
\begin{aligned}
\Tr_{\text{(qudit)}} \big(\ket\psi \bra\psi\big) &= \sum_{i,j,m=1}^{d} c_i \overline{c_j} \braket{m|i}\braket{j|m}\ket{\two{i}}_d\bra{\two{j}}\\
&= \sum_{i=1}^d |c_i|^2 \ket{\two{i}}_d\bra{\two{i}}\\
&\eqqcolon \rho'\\
\Tr_{(n \text{ leftmost qubits)}} (\rho') &= \sum_{j=1}^{n} |c_j|^2 \ket{00\ldots 0}\bra{00\ldots 0} +\\
&\quad +\sum_{i=n+1}^{d} |c_i|^2 \ket{\underline{2^{d-i}}}_{d-n}\bra{\underline{2^{d-i}}}\\
&\eqqcolon \rhon
\end{aligned}
\end{equation}
where we have introduced a subscript in the notation $\ket{{2^{d-i}}}_{d-n}$ to indicate that there are $d-n$ qubits left over.
Notice that the only possible nonzero eigenvalues of $\rhon$ are given by: $\sum_{i=1}^{n} |c_i|^2, |c_{n+1}|^2, |c_{n+2}|^2, \ldots, |c_{d}|^2$. Since the coefficients $|c_i|$ are arranged in non-increasing order by assumption, it follows that
\begin{equation}
\max_{\ket\phin \in \mathcal{P}^{(k)}} |\braket{\phin | \psi}|^2 \leq \lambda_{\max}\left(\rhon\right) = \sum_{i=1}^{n} |c_i|^2 \leq \sum_{j=1}^{k-1} |c_i|^2
\end{equation}
as required.
\end{proof}

The above Lemma finally allows us to prove the equality between the geometric quantifiers for $k$-coherence and $k+1$-partite entanglement in the first step of the conversion protocol, which is the last main result of this paper.
\begin{theorem}
\label{thm:geometric}
Given $\rho$ and the transformed state $\rho'$ as described in the conversion protocol, we have
\begin{equation}
C_G^{(k)} (\rho) = E_G^{(k+1)} (\rho').
\end{equation}
\end{theorem}
\begin{proof}
The proof follows the methods of Refs.~\cite{streltsov_2010} and \cite{streltsov_2015}.

By Thm.~\ref{thm:distance_based_relations} and Prop.~\ref{prop:geometric_fidelity}, this result amounts to showing that there exists a state $\chi \in \mathcal{B}^{(k)}$ which is the closest $k$-producible state to $\rho'$ with respect to the distance given by $D(\rho,\sigma) = 1 - F(\rho,\sigma)$. In other words, we need to find $\chi \in \mathcal{B}^{(k)}$ s.t. $E_G^{(k+1)}(\rho') = 1 - F(\rho', \chi)$.

First, let $\{p_i, \ket{\phi_i}\}$ be the optimal convex decomposition of $\rho'$ which realizes the infimum in the convex roof extension of $E_G^{(k+1)}$, that is,
\begin{equation}
\label{eq:optimal_decomp}
E_G^{(k+1)} (\rho') = \sum_i p_i \, E_G^{(k+1)}(\ket{\phi_i}).
\end{equation}
Define $\{\ket{\xi_i}\} \in \mathcal{P}^{(k)}$ to be the closest $k$-producible states to each of the states $\{\ket{\phi_i}\}$ with respect to the fidelity-based distance, that is,
\begin{equation}
\label{eq:optimal_fidelity}
E_G^{(k+1)} (\ket{\phi_i}) = 1 - F(\ket{\xi_i}, \ket{\phi_i}) \;\;\;\forall\, i.
\end{equation}

Now, notice that the support of $\rho'$ is spanned by states of the form
\begin{equation}
\ket{\psi'} = \sum_i p_i \, \ket{i} \ket{\underline{2^{d-i}}} \; \in \mathcal{B}^{(k)},
\end{equation}
which means that any pure state in a convex decomposition of $\rho'$ can be expressed as a complex linear combination of such $\ket{\psi'}$. In particular, each of the states $\{\ket{\phi_i}\}$ can be written as
\begin{equation}
\label{eq:expansion_purestate}
\ket{\phi_i} = \sum_{j} r_j \, \ket{j} \ket{\underline{2^{d-j}}}.
\end{equation}
By Lemma \ref{lem:purefidelity}, the closest $k$-producible state to each $\ket{\phi_i}$ is a state $\ket{\xi_i} \in \mathcal{B}^{(k)}$, and so we have
\begin{equation}
\ket{\xi_i} = \sum_{j=1}^{k-1} e_j \ket{j} \ket{\underline{2^{d-j}}}.
\end{equation}
Now, define
\begin{equation}
\chi = \sum_j q_j \ket{\xi_j}\bra{\xi_j}
\end{equation}
with coefficients
\begin{equation}
q_j = p_j \frac{1 - E_G^{(k+1)}\left(\ket{\phi_j}\right)}{1 - E_G^{(k+1)}\left(\rho'\right)}.
\end{equation}
Since each $\ket{\xi_i}$ is in $\mathcal{B}^{(k)}$, we get that $\chi \in \mathcal{B}^{(k)}$. Recalling that $\sqrt{F}$ satisfies the so-called strong joint concavity property \cite{nielsen_2011}, defined as
\begin{equation}
\begin{aligned}
&\sqrt{F\left( \sum_i a_i \ket{\alpha_i}\bra{\alpha_i},\, \sum_i b_i \ket{\beta_i}\bra{\beta_i} \right)} \\&\geq \sum_i \sqrt{a_i\,b_i\,F\left( \ket{\alpha_i}, \ket{\beta_i} \right)},
\end{aligned}
\end{equation}
we get
\begin{equation}\begin{aligned}
\sqrt{F(\rho', \chi)} &\geq \sum_j \sqrt{p_j\,q_j\,F\left(\ket{\phi_j}, \ket{\xi_j}\right)}\\
&= \sum_j \sqrt{p_j^2 \frac{F\left(\ket{\phi_j}, \ket{\xi_j}\right)^2}{1 - E_G^{(k+1)}\left(\rho'\right)}}\\
&= \sqrt{1-E_G^{(k+1)}\left(\rho'\right)}
\end{aligned}
\end{equation}
which follows by Eqs.~\eqref{eq:optimal_fidelity} and \eqref{eq:optimal_decomp}. This gives
\begin{equation}
E_G^{(k+1)}\left(\rho'\right) \geq 1 - F(\rho', \chi)
\end{equation}
and since $\chi \in \mathcal{P}^{(k)}$ is a convex combination of $k$-producible pure states, by Prop.~\ref{prop:geometric_fidelity} we also have
\begin{equation}
E_G^{(k+1)}\left(\rho'\right) \leq 1 - F(\rho', \chi)
\end{equation}
and so $\chi$ is the closest $k$-producible state to $\rho'$. But since $\chi \in \mathcal{B}^{(k)}$ by construction, Thm.~\ref{thm:distance_based_relations} gives
\begin{equation}
C_G^{(k)} (\rho) = E_G^{(k+1)} (\rho')
\end{equation}
as required.
\end{proof}

The above result has implications for the quantification of $k+1$-partite entanglement, since for any state of the form $\rho'$ as obtained from the conversion protocol in Thm.~\ref{thm:coherence_pure}, the entanglement can be quantified by considering the quantification of $k$-coherence instead. It has been recently shown that optimization over sets of $k$-coherent states can be expressed as efficiently computable semidefinite programs (SDPs) \cite{ringbauer_2017}, and together with the fact that the computation of the fidelity function can be cast as an SDP as well \cite{watrous_2009,watrous_2013}, we have that the geometric measure of $k+1$-partite entanglement of any state $\rho'$ can be quantified by the following SDP:
\begin{equation}\begin{aligned}
\sqrt{1-E_G^{(k+1)}\left(\rho'\right)} =\; &\text{max}\hspace{-5pt}&& \Re\Tr(X)\\
&\text{s.t.}&& \begin{pmatrix}\rho' & X\\X^\dagger & \displaystyle\sum_{I \in \mathscr{Q}_{k-1}} P_I Y_I P_I\end{pmatrix} \geq 0\\
&&&\Tr\left(\sum_{I \in \mathscr{Q}_{k-1}} P_I Y_I P_I\right) = 1\\
&&&Y_I \geq 0 \;\, \forall \,I \in \mathscr{Q}_{k-1}\\
&&&X \in \mathbb{C}^{d \times d}
\end{aligned}\end{equation}
where $\mathscr{Q}_{k-1}$ denotes the set of all $k$-element combinations from $\{1, 2, \ldots, d\}$ and $P_I$ denotes the orthogonal projection $P_I = \sum_{i\in I} \ket{i}\bra{i}$ \cite{ringbauer_2017}.

Moreover, in a very similar way to the proof of Lemma \ref{lem:purefidelity}, one can derive a closed formula for $C_{G}^{(k)}$ of arbitrary pure states $\ket{\psi}$ as
\begin{equation}
  C_G^{(k)}(\ket\psi) = 1 - \sum_{i=1}^{k-1} |c^\downarrow_i|^2
\end{equation}
where $c^\downarrow_i$ denotes the $i$th largest coefficient (by absolute value) of $\ket\psi$. This entails a closed formula for $E_{G}^{(k+1)}$ of the corresponding states $\ket{\Psi'}$ at the output of the conversion protocol. For completeness, we present a full proof of this fact in the Appendix~\ref{sec:appendix}.

We note that quantitative relations can also be obtained for other measures of coherence and entanglement based on the convex roof, similarly to the cases in \cite{streltsov_2015,yuan_2015,winter_2016,chin_2017}. Such monotones are built by taking suitable functions defined on pure states \cite{vidal_2000,du_2015} and extending them to mixed states by minimizing over all pure-state decompositions \cite{uhlmann_1998}. Since the conversion of $k$-coherence into multipartite entanglement is isometric, there is a one-to-one correspondence between such decompositions for input and output states, and close relations between equivalent measures can be derived \cite{zhu_2017}.


\section{Conclusions}

We have investigated the relation between the nonclassicality (in the form of superposition) of a single quantum system and the genuine multipartite entanglement which can be obtained from it in physical processes. We have shown that a faithful conversion of multilevel nonclassicality into multipartite entanglement is always possible by mapping superpositions between $k$ levels of a system into entanglement between the system and $k$ ancillas via unitary operations. As an explicit implementation of this result, we presented a reversible protocol for the conversion of $k$-coherence into genuine multipartite entanglement, showing that the strength of the final entanglement among all parties is bounded by the initial amount of quantum coherence, and can in fact be exactly equivalent under a suitable choice of geometric quantifiers.

This reveals a qualitative and quantitative connection between multilevel nonclassicality and multipartite entanglement, generalizing previous results in the resource theory of quantum coherence \cite{baumgratz_2014,streltsov_2015,streltsov_2016}, and further contributing towards the formalization of nonclassicality as a resource \cite{killoran_2016,theurer_2017,mukhopadhyay_2017,regula_2018,tan_2017}. In particular, multilevel coherence and multipartite entanglement provide significant operational advantages over the resources of standard quantum coherence and bipartite entanglement \cite{horodecki_2009,walter_2016,ringbauer_2017} and are key ingredients for practical applications such as quantum computation, quantum networks, sensing, and metrology \cite{kimble_2008,braun_2017,pezze_2016,walter_2016}. By providing constructive schemes for their interchange in compliance with the respective resource theories, our work lays the foundation for a complete characterization of the interrelations between the two fundamental resources, and may further serve as an inspiration for novel hybrid approaches to quantum technologies.

\begin{acknowledgments}

We acknowledge financial support from the European Research Council (ERC) under the Starting Grant GQCOP (Grant No.~637352), the European Union's Horizon 2020 Research and Innovation Programme under the Marie Sk\l{}odowska-Curie Action OPERACQC (Grant Agreement No.~661338), the Foundational Questions
Institute under the Physics of the Observer Programme (Grant No.~FQXi-RFP-1601), and the National Science
Center in Poland (POLONEZ Grant No.~UMO-2016/21/P/ST2/04054).

 \textit{Note.} --- During completion of this work, an independent investigation of $k$-coherence and its relation to bipartite entanglement of Schmidt rank $k$ was presented by S.~Chin \cite{chin_2017,chin_2017-1}.
\end{acknowledgments}

\bibliography{main}

\appendix 
\section{}\label{sec:appendix}

\begin{proposition}
Given an orthonormal basis $\{\ket{i}\}$, for any pure state $\ket\psi = \sum_i c_i \ket{i}$ it holds that
\begin{equation}
C_G^{(k)}(\ket\psi) = 1 - \sum_{i=1}^{k-1} |c^\downarrow_i|^2
\end{equation}
where $c^\downarrow_i$ denotes the $i$th largest coefficient (by absolute value) of $\ket\psi$.
\end{proposition}
\begin{proof}
A general state $\ket\eta \in \incoh{k-1}$ is given by
\begin{equation}
\ket\eta = \sum_{j \in \mathcal{I}} d_j \ket{j}
\end{equation}
where $\mathcal{I}$ is a subset of $k-1$ elements of the indices $\{1,\ldots,d\}$. Hence
\begin{equation}
\begin{aligned}
F(\ket\psi, \ket\eta) &= |\braket{\psi | \eta}|^2\\
&=  \left| \sum_{i=1}^{d} \sum_{j \in \mathcal{I}} \overline{c_i} d_j \braket{i|j} \right|^2\\
&= \left| \sum_{j\in \mathcal{I}} \overline{c_j} d_j \right|^2\\
&\leq \sum_{j\in \mathcal{I}} |c_j|^2 \sum_{j\in \mathcal{I}} |d_j|^2\\
&= \sum_{j\in \mathcal{I}} |c_j|^2\\
&\leq \sum_{i=1}^{k-1} |c^\downarrow_i|^2
\end{aligned}
\end{equation}
using the Cauchy-Schwarz inequality and the fact that the choice of $k-1$ largest coefficients $|c_i|$ maximizes the expression. The bound is tight, since one can always reach it by considering the state
\begin{equation}
\ket{\psi^c} = \frac{\sum_{i=1}^{k-1} c^\downarrow_i \ket{i^\downarrow}}{\sqrt{\sum_{i=1}^{k-1} |c^\downarrow_i|^2}} \; \in \incoh{k-1}
\end{equation}
where $\ket{i^\downarrow}$ are the basis vectors corresponding to the coefficients $c^\downarrow_i$ of $\ket\psi$. Therefore we have
\begin{equation}
C_G^{(k)}(\ket\psi) = 1 - \sup_{\ket{\eta} \in \incoh{k-1}} F(\ket\psi, \ket\eta) = 1 - \sum_{i=1}^{k-1} |c^\downarrow_i|^2.
\end{equation}
\end{proof}

\end{document}